\documentclass[conference,letterpaper]{IEEEtran}

\usepackage[top=0.5in, bottom=0.55in, left=0.6in, right=0.6in]{geometry}
\usepackage[utf8]{inputenc} 
\usepackage[T1]{fontenc}

\ifCLASSINFOpdf

\else

\fi
\usepackage{amsmath}
\usepackage{ifthen}
\usepackage{cite}
\usepackage{amsthm,amsfonts,amssymb}
\usepackage{graphicx}
\usepackage{subfig}
\usepackage{blkarray}
\usepackage[mathscr]{euscript}
\usepackage{enumitem}
\usepackage{algorithm}
\usepackage[noend]{algpseudocode}
\usepackage{blkarray}
\usepackage{stfloats}%
\newtheorem{theorem}{Theorem}[section]
\newtheorem{corollary}{Corollary}[section]
\newtheorem{proposition}{Proposition}[section]
\newtheorem{lemma}{Lemma}[section]
\theoremstyle{remark}
\newtheorem*{remark}{Remark}
\theoremstyle{definition}
\newtheorem{definition}{Definition}[section]
\usepackage{url}
\usepackage{arydshln}
\usepackage{mathtools}
\usepackage{dsfont}
\usepackage{physics}
\usepackage{tikz}
\usepackage{mathdots}
\usepackage{cancel}
\usepackage{color}
\usepackage{siunitx}
\usepackage{array}
\usepackage{multirow}
\usepackage{amssymb}
\usepackage{gensymb}
\usepackage{tabularx}
\usepackage{extarrows}
\usepackage{booktabs}
\usetikzlibrary{fadings}
\usetikzlibrary{patterns}
\usetikzlibrary{shadows.blur}
\usetikzlibrary{shapes}

\usepackage[cmintegrals]{newtxmath}

\begin{document}
	\count\footins = 1000
%
\title{A Feedback Capacity-Achieving Coding Scheme for the $(d,\infty)$-RLL Input-Constrained Binary Erasure Channel}
%
%
%



\author{V.~Arvind Rameshwar \ and \ Navin Kashyap
	\thanks{The authors are with the Department of Electrical Communication Engineering, Indian Institute of Science, Bengaluru 560012. Email: \{\texttt{vrameshwar}, \texttt{nkashyap}\}\texttt{@iisc.ac.in}}
	\thanks{This work was supported in part by a Qualcomm Innovation Fellowship India 2020. The work of V.~A.~Rameshwar was supported by a Prime Minister's Research Fellowship, from the Ministry of Education, Govt. of India. The work of N.~Kashyap was supported in part by MATRICS grant \ MTR/2017/000368 from the Science and Engineering Research Board (SERB), Govt. of India.}
}
\IEEEoverridecommandlockouts
%



\maketitle

\begin{abstract}
This paper considers the memoryless input-constrained binary erasure channel (BEC). The channel input constraint is the $(d,\infty)$-runlength limited (RLL) constraint, which mandates that any pair of successive $1$s in the input sequence be separated by at least $d$ $0$s. 
We consider a scenario where there is causal, noiseless feedback from the decoder. We demonstrate a simple, labelling-based, zero-error feedback coding scheme, which we prove to be feedback capacity-achieving, and, as a by-product, obtain an explicit characterization of the feedback capacity. Our proof is based on showing that the rate of our feedback coding scheme equals an upper bound on the feedback capacity derived using the single-letter bounding techniques of Sabag et al. (2017). 
Further, we note using the tools of Thangaraj (2017) that there is a gap between the feedback and non-feedback capacities of the $(d,\infty)$-RLL input constrained BEC, {at least for $d=1,2$}.

\end{abstract}



\section{Introduction}
\label{sec:intro}
In most data {storage systems}, constrained coding is used to prohibit the appearance of certain input sequences that are more error-prone than others. Here, arbitrary user data sequences are encoded into sequences that respect the constraint (see, for example, \cite{Roth} or \cite{Immink}). In this work, we consider the problem of designing good binary constrained codes for use over a binary erasure channel (BEC), with erasure probability $\epsilon$ (see Figure 1), {when the channel output is fed back noiselessly to the transmitter}.


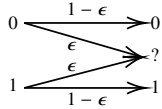
\begin{figure}[!h]
					\label{fig:bec}
	\centering
		\resizebox{0.13\textwidth}{!}{

			\tikzset{every picture/.style={line width=0.75pt}} 
			
			\begin{tikzpicture}[x=0.75pt,y=0.75pt,yscale=-1,xscale=1]
				
				\draw    (454.24,136.67) -- (521.85,136.67) ;
				\draw [shift={(523.85,136.67)}, rotate = 180] [color={rgb, 255:red, 0; green, 0; blue, 0 }  ][line width=0.75]    (10.93,-3.29) .. controls (6.95,-1.4) and (3.31,-0.3) .. (0,0) .. controls (3.31,0.3) and (6.95,1.4) .. (10.93,3.29)   ;
				\draw    (454.24,136.67) -- (519.94,155.67) ;
				\draw [shift={(521.86,156.22)}, rotate = 196.13] [color={rgb, 255:red, 0; green, 0; blue, 0 }  ][line width=0.75]    (10.93,-3.29) .. controls (6.95,-1.4) and (3.31,-0.3) .. (0,0) .. controls (3.31,0.3) and (6.95,1.4) .. (10.93,3.29)   ;
				\draw    (454.24,173.79) -- (521.85,173.79) ;
				\draw [shift={(523.85,173.79)}, rotate = 180] [color={rgb, 255:red, 0; green, 0; blue, 0 }  ][line width=0.75]    (10.93,-3.29) .. controls (6.95,-1.4) and (3.31,-0.3) .. (0,0) .. controls (3.31,0.3) and (6.95,1.4) .. (10.93,3.29)   ;
				\draw    (454.24,173.79) -- (519.93,156.72) ;
				\draw [shift={(521.86,156.22)}, rotate = 165.43] [color={rgb, 255:red, 0; green, 0; blue, 0 }  ][line width=0.75]    (10.93,-3.29) .. controls (6.95,-1.4) and (3.31,-0.3) .. (0,0) .. controls (3.31,0.3) and (6.95,1.4) .. (10.93,3.29)   ;
				
				\draw (476.48,146) node [anchor=north west][inner sep=0.75pt]  [font=\footnotesize]  {$\epsilon $};
				\draw (476.48,158) node [anchor=north west][inner sep=0.75pt]  [font=\footnotesize]  {$\epsilon $};
				\draw (477.13,123.4) node [anchor=north west][inner sep=0.75pt]  [font=\footnotesize]  {$1-\epsilon $};
				\draw (477.13,175) node [anchor=north west][inner sep=0.75pt]  [font=\footnotesize]  {$1-\epsilon $};
				\draw (443.86,131.32) node [anchor=north west][inner sep=0.75pt]  [font=\footnotesize]  {${\displaystyle 0}$};
				\draw (443.86,166.96) node [anchor=north west][inner sep=0.75pt]  [font=\footnotesize]  {$1$};
				\draw (525.01,131.32) node [anchor=north west][inner sep=0.75pt]  [font=\footnotesize]  {$0$};
				\draw (525.01,149.88) node [anchor=north west][inner sep=0.75pt]  [font=\footnotesize]  {$?$};
				\draw (525.01,167.7) node [anchor=north west][inner sep=0.75pt]  [font=\footnotesize]  {$1$};

			\end{tikzpicture}

		}

		\caption{The binary erasure channel (BEC) with erasure probability $\epsilon$.}
		
\end{figure}

Our focus is on the $(d,\infty)$-RLL input constraint, which mandates that there are at least $d$ $0$s between every pair of successive $1$s in the input sequence. Figure 2 shows a state transition graph that represents the constraint. The $(d,\infty)$-RLL constraint is a special case of the $(d,k)$-RLL constraint, which admits only binary sequences with at least $d$ and at most $k$ $0$s between successive $1$s. 
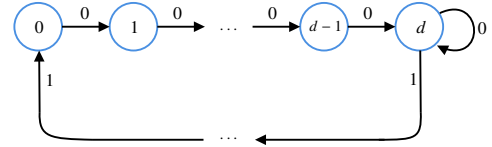
\begin{figure}[!h]
	
	\begin{center}
		\resizebox{0.35\textwidth}{!}{
			\tikzset{every picture/.style={line width=0.75pt}} 
			
			\begin{tikzpicture}[x=0.75pt,y=0.75pt,yscale=-1,xscale=1]
				
				
				\draw  [color={rgb, 255:red, 74; green, 144; blue, 226 }  ,draw opacity=1 ][line width=1.5]  (130,96) .. controls (130,82.19) and (141.19,71) .. (155,71) .. controls (168.81,71) and (180,82.19) .. (180,96) .. controls (180,109.81) and (168.81,121) .. (155,121) .. controls (141.19,121) and (130,109.81) .. (130,96) -- cycle ;
				\draw  [color={rgb, 255:red, 74; green, 144; blue, 226 }  ,draw opacity=1 ][line width=1.5]  (230,96) .. controls (230,82.19) and (241.19,71) .. (255,71) .. controls (268.81,71) and (280,82.19) .. (280,96) .. controls (280,109.81) and (268.81,121) .. (255,121) .. controls (241.19,121) and (230,109.81) .. (230,96) -- cycle ;
				\draw  [color={rgb, 255:red, 74; green, 144; blue, 226 }  ,draw opacity=1 ][line width=1.5]  (430,96) .. controls (430,82.19) and (441.19,71) .. (455,71) .. controls (468.81,71) and (480,82.19) .. (480,96) .. controls (480,109.81) and (468.81,121) .. (455,121) .. controls (441.19,121) and (430,109.81) .. (430,96) -- cycle ;
				\draw  [color={rgb, 255:red, 74; green, 144; blue, 226 }  ,draw opacity=1 ][line width=1.5]  (530,96) .. controls (530,82.19) and (541.19,71) .. (555,71) .. controls (568.81,71) and (580,82.19) .. (580,96) .. controls (580,109.81) and (568.81,121) .. (555,121) .. controls (541.19,121) and (530,109.81) .. (530,96) -- cycle ;
				\draw [line width=1.5]    (180,96) -- (226,96) ;
				\draw [shift={(230,96)}, rotate = 180] [fill={rgb, 255:red, 0; green, 0; blue, 0 }  ][line width=0.08]  [draw opacity=0] (11.61,-5.58) -- (0,0) -- (11.61,5.58) -- cycle    ;
				\draw [line width=1.5]    (280,96) -- (326,96) ;
				\draw [shift={(330,96)}, rotate = 180] [fill={rgb, 255:red, 0; green, 0; blue, 0 }  ][line width=0.08]  [draw opacity=0] (11.61,-5.58) -- (0,0) -- (11.61,5.58) -- cycle    ;
				\draw [line width=1.5]    (380,96) -- (426,96) ;
				\draw [shift={(430,96)}, rotate = 180] [fill={rgb, 255:red, 0; green, 0; blue, 0 }  ][line width=0.08]  [draw opacity=0] (11.61,-5.58) -- (0,0) -- (11.61,5.58) -- cycle    ;
				\draw [line width=1.5]    (480,96) -- (526,96) ;
				\draw [shift={(530,96)}, rotate = 180] [fill={rgb, 255:red, 0; green, 0; blue, 0 }  ][line width=0.08]  [draw opacity=0] (11.61,-5.58) -- (0,0) -- (11.61,5.58) -- cycle    ;
				\draw [line width=1.5]    (556,121) -- (556.5,185.73) ;
				\draw [line width=1.5]    (156.03,125) -- (156.5,185.73) ;
				\draw [shift={(156,121)}, rotate = 89.56] [fill={rgb, 255:red, 0; green, 0; blue, 0 }  ][line width=0.08]  [draw opacity=0] (11.61,-5.58) -- (0,0) -- (11.61,5.58) -- cycle    ;
				\draw [line width=1.5]    (156.5,185.73) .. controls (155.5,212.73) and (167.5,213.73) .. (211.5,214.73) ;
				\draw [line width=1.5]    (556.5,185.73) .. controls (557.5,213.73) and (551.5,214.73) .. (512.5,214.73) ;
				\draw [line width=1.5]    (211.5,214.73) -- (328.5,214.73) ;
				\draw [line width=1.5]    (386.5,214.73) -- (518.5,214.73) ;
				\draw [shift={(382.5,214.73)}, rotate = 0] [fill={rgb, 255:red, 0; green, 0; blue, 0 }  ][line width=0.08]  [draw opacity=0] (11.61,-5.58) -- (0,0) -- (11.61,5.58) -- cycle    ;
				\draw [line width=1.5]    (576.5,81.73) .. controls (622.56,60.17) and (626.36,138.5) .. (576.61,117.17) ;
				\draw [shift={(573.5,115.73)}, rotate = 386.13] [fill={rgb, 255:red, 0; green, 0; blue, 0 }  ][line width=0.08]  [draw opacity=0] (11.61,-5.58) -- (0,0) -- (11.61,5.58) -- cycle    ;
				
				\draw (344,96) node [anchor=north west][inner sep=0.75pt]  [font=\large]  {$\dotsc $};
				\draw (344,210.7) node [anchor=north west][inner sep=0.75pt]  [font=\large]  {$\dotsc $};
				\draw (150,88.4) node [anchor=north west][inner sep=0.75pt] [font=\Large]   {$0$};
				\draw (250,88.4) node [anchor=north west][inner sep=0.75pt]   [font=\Large] {$1$};
				\draw (438,88.4) node [anchor=north west][inner sep=0.75pt]   [font=\large] {$d-1$};
				\draw (550,88.4) node [anchor=north west][inner sep=0.75pt]  [font=\Large]  {$d$};
				\draw (198,73.4) node [anchor=north west][inner sep=0.75pt]  [font=\Large]  {$0$};
				\draw (295,73.4) node [anchor=north west][inner sep=0.75pt]  [font=\Large]  {$0$};
				\draw (395,73.4) node [anchor=north west][inner sep=0.75pt]  [font=\Large]  {$0$};
				\draw (495,73.4) node [anchor=north west][inner sep=0.75pt]  [font=\Large]  {$0$};
				\draw (615,89.4) node [anchor=north west][inner sep=0.75pt] [font=\Large]   {$0$};
				\draw (543,143.4) node [anchor=north west][inner sep=0.75pt]  [font=\Large]  {$1$};
				\draw (161,144.4) node [anchor=north west][inner sep=0.75pt] [font=\Large]   {$1$};

			\end{tikzpicture}
		}
	\end{center}
	\label{fig:d_inf}
	\caption{The state transition graph for the $(d,\infty)$-RLL constraint.}
\end{figure}
The system model under investigation in this paper is the classical discrete memoryless channel (DMC) (introduced by Shannon in \cite{Sh48}) in the presence of causal, noiseless feedback, with inputs that are constrained to obey the $(d,\infty)$-RLL constraint. Input-constrained DMCs fall under the broad class of discrete finite-state channels (DFSCs, or FSCs). Figure 3 shows a generic constrained DMC with decoder feedback. It is well-known from Shannon's work in \cite{Sh56} that feedback does not increase the capacity of an unconstrained DMC, i.e., $C_{\text{DMC}}^{\text{fb}} = C_{\text{DMC}} = \sup_{P(x)} I(X;Y)$. However, Shannon's arguments do not apply to the case of an input-constrained DMC, and we shall employ special tools to determine the feedback capacity in this setting.


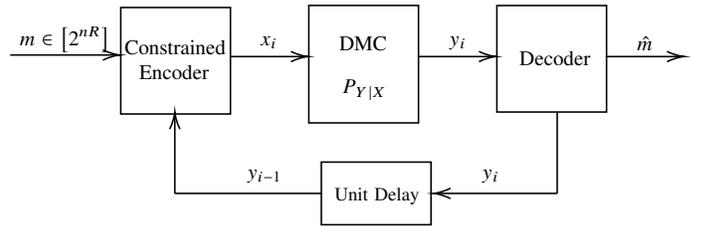
\begin{figure}[!h]
	
	\centering
		\resizebox{0.5\textwidth}{!}{

			\tikzset{every picture/.style={line width=0.75pt}} 
			
			\begin{tikzpicture}[x=0.75pt,y=0.75pt,yscale=-1,xscale=1]
				
				\draw   (341,126.32) -- (411,126.32) -- (411,194.32) -- (341,194.32) -- cycle ;
				\draw    (410.5,157.65) -- (458.5,157.65) ;
				\draw [shift={(460.5,157.65)}, rotate = 180] [color={rgb, 255:red, 0; green, 0; blue, 0 }  ][line width=0.75]    (10.93,-3.29) .. controls (6.95,-1.4) and (3.31,-0.3) .. (0,0) .. controls (3.31,0.3) and (6.95,1.4) .. (10.93,3.29)   ;
				\draw   (461,125) -- (531,125) -- (531,200) -- (461,200) -- cycle ;
				\draw   (581,125.06) -- (651,125.06) -- (651,193.06) -- (581,193.06) -- cycle ;
				\draw    (530.5,157.65) -- (578.5,157.65) ;
				\draw [shift={(580.5,157.65)}, rotate = 180] [color={rgb, 255:red, 0; green, 0; blue, 0 }  ][line width=0.75]    (10.93,-3.29) .. controls (6.95,-1.4) and (3.31,-0.3) .. (0,0) .. controls (3.31,0.3) and (6.95,1.4) .. (10.93,3.29)   ;
				\draw    (650.5,157.65) -- (698.5,157.65) ;
				\draw [shift={(700.5,157.65)}, rotate = 180] [color={rgb, 255:red, 0; green, 0; blue, 0 }  ][line width=0.75]    (10.93,-3.29) .. controls (6.95,-1.4) and (3.31,-0.3) .. (0,0) .. controls (3.31,0.3) and (6.95,1.4) .. (10.93,3.29)   ;
				\draw    (270.5,156.65) -- (338.5,156.65) ;
				\draw [shift={(340.5,156.65)}, rotate = 180] [color={rgb, 255:red, 0; green, 0; blue, 0 }  ][line width=0.75]    (10.93,-3.29) .. controls (6.95,-1.4) and (3.31,-0.3) .. (0,0) .. controls (3.31,0.3) and (6.95,1.4) .. (10.93,3.29)   ;
				
				\draw    (619.5,193.04) -- (619.5,245.04) ;
				\draw    (542.5,245.04) -- (619.5,245.04) ;
				\draw [shift={(540.5,245.04)}, rotate = 0] [color={rgb, 255:red, 0; green, 0; blue, 0 }  ][line width=0.75]    (10.93,-3.29) .. controls (6.95,-1.4) and (3.31,-0.3) .. (0,0) .. controls (3.31,0.3) and (6.95,1.4) .. (10.93,3.29)   ;
				\draw   (469,225.04) -- (539,225.04) -- (539,265.04) -- (469,265.04) -- cycle ;
				
				\draw    (375.5,195.04) -- (375.5,245.04) ;
				\draw [shift={(375.5,193.04)}, rotate = 90] [color={rgb, 255:red, 0; green, 0; blue, 0 }  ][line width=0.75]    (10.93,-3.29) .. controls (6.95,-1.4) and (3.31,-0.3) .. (0,0) .. controls (3.31,0.3) and (6.95,1.4) .. (10.93,3.29)   ;
				\draw    (375.5,245.04) -- (469.5,245.04) ;
				
				\draw (481,171.01) node [anchor=north west][inner sep=0.75pt]  [font=\normalsize]  {$P_{Y|X}$};
				\draw (669,143.01) node [anchor=north west][inner sep=0.75pt]  [font=\normalsize]  {$\hat{m}$};
				\draw (549,143.01) node [anchor=north west][inner sep=0.75pt]  [font=\normalsize]  {$y_{i}$};
				\draw (429,143.01) node [anchor=north west][inner sep=0.75pt]  [font=\normalsize]  {$x_{i}$};
				\draw (594,152.61) node [anchor=north west][inner sep=0.75pt]  [font=\normalsize] [align=left] {Decoder};
				\draw (342,145.65) node [anchor=north west][inner sep=0.75pt]   [align=left] {{\normalsize Constrained}\\{\normalsize  \ \ Encoder}};
				\draw (479,143) node [anchor=north west][inner sep=0.75pt]   [align=left] {DMC};
				\draw (275,137.01) node [anchor=north west][inner sep=0.75pt]  [font=\normalsize]  {${\displaystyle m\in \left[ 2^{nR}\right]}$};
				\draw (420,228.4) node [anchor=north west][inner sep=0.75pt]  [font=\normalsize]  {$y_{i-1}$};
				\draw (570,228.4) node [anchor=north west][inner sep=0.75pt]  [font=\normalsize]  {$y_{i}$};
				\draw (476,240) node [anchor=north west][inner sep=0.75pt]  [font=\small] [align=left] {Unit Delay};

			\end{tikzpicture}
}
	\label{fig:genconstdmcfb}
\caption{The system model of a generic input-constrained DMC with causal, noiseless feedback.}
\end{figure}

In this work, we provide a simple, labelling-based, zero-error feedback coding scheme for the $(d,\infty)$-RLL input-constrained BEC, similar to the those presented in \cite{SabBEC} and \cite{0k}, for other input-constrained BECs.
We then prove that our feedback coding scheme is in fact feedback capacity-achieving. Our method uses the single-letter bounding techniques in \cite{Single} to obtain an upper bound on feedback capacity, which we show to be equal to the rate of our proposed coding scheme. As a result, we are able to explicitly characterize the feedback capacity of the $(d,\infty)$-RLL input-constrained BEC, which is given by a $(d+1)$-parameter optimization problem:
\[
C^{\text{fb}}_{(d,\infty)}(\epsilon) = \max_{\substack{{\delta_0,\ldots,\delta_d\geq 0}\\ {\sum_{i=0}^{d} \delta_i\leq 1}}} \frac{\bar{\epsilon}\left(\sum\limits_{i=0}^{d}\epsilon^ih_b(\delta_i)\right)}{\sum\limits_{i=0}^{d}\epsilon^i + d\bar{\epsilon}\left(\sum\limits_{i=0}^{d}\epsilon^i\delta_i\right)}.
\]
Our formula generalizes the coding scheme in \cite{SabBEC}, where the feedback capacity of the $(1,\infty)$-RLL input-constrained BEC was derived using dynamic programming techniques. Our work supplements the results in \cite{0k}, which provided the feedback capacity of the $(0,k)$-RLL input-constrained BEC. However, interestingly, unlike the previous two results, the feedback capacity of the $(d,\infty)$-RLL input-constrained BEC, for general $d$, does not equal the capacity when the encoder possesses non-causal knowledge of erasures (this observation was also made in \cite{0k}).

The remainder of the paper is organized as follows: Section \ref{sec:notation} introduces the notation and refreshes some preliminary background. Section \ref{sec:main} states our main results. The optimal feedback coding scheme and an upper bound on the feedback capacity that is equal to the rate of the coding scheme, are discussed in Section \ref{sec:fbcode}. The feedback capacity is thus explicitly derived in this section. 
Finally, Section \ref{sec:conclusion} contains concluding remarks and a discussion on possible future work.

\section{Notation and Preliminaries}
\label{sec:notation}
\subsection{Notation}

Random variables will be denoted by capital letters, and their realizations by lower-case letters, e.g., $X$ and $x$, respectively. Calligraphic letters, e.g., $\mathscr{X}$, denote sets. The notation $[n]$ denotes the set, $\{1,2,\ldots,n\}$, of integers.
The notations $X^{N}$ and $x^N$ denote the random vector $(X_1,\ldots,X_N)$ and the realization $(x_1,\ldots,x_N)$, respectively. Further, $P(x), P(y)$ and $P(y|x)$ are used to denote the probabilities $P_X(x), P_Y(y)$ and $P_{Y|X}(y|x)$, respectively. Also, the notations $H(X)$ (or $H(P_X)$) and $I(X;Y)$ stand for the entropy of $X$, and the mutual information between $X$ and $Y$, respectively, and $h_b(p):=-p\log_2 p - (1-p)\log_2(1-p)$ is the binary entropy function, for $p\in [0,1]$. 
Finally, for ease of reading, we define $\bar{\alpha}=1-\alpha$. 
All logarithms are to the base $2$.

\subsection{Problem Definition}

The communication setting of an input-constrained memoryless channel with causal, noiseless feedback is shown in Figure 3. A message $M$ is drawn uniformly from the set $\{1,2,\ldots,2^{nR}\}$, and is made available to the constrained encoder. The constrained encoder at time $i$ has, in addition to the message, access to noiseless feedback in the form of the outputs, $y^{i-1}$, from the decoder. It then produces a binary input symbol $x_i \in \{0,1\}$, as a function of the specific instance of the message, $m$, and the  outputs, $y^{i-1}$, in such a manner that the $(d,\infty)$-RLL constraint (see Figure 2) is respected. We set the channel state alphabet, $\mathscr{S}$, to be $\{0,1,\ldots,d\}$. The channel is memoryless in the sense that $P(y_i|x^{i},y^{i-1}) = P(y_i|x_i)$, for all $i$.

\begin{definition}
	An $(n,2^{nR},(d,\infty))$ \emph{feedback} code for an input-constrained channel  is defined by the encoding functions:
	\begin{equation*}
		f_i: \{1,\ldots, 2^{nR}\}\times \mathscr{Y}^{i-1} \rightarrow \mathscr{X}, \quad i\in [n],
	\end{equation*}
	which satisfy $f_i(m,y^{i-1}) = 0$, if $f_{(i-j)^+}(m, y^{(i-j-1)^+}) = 1$ (where $x^+$ is equal to   $\max\{x,0\}$), for some $j\in [d]$, and a decoding function $\Gamma: \mathscr{Y}^n \rightarrow \{1,\ldots,2^{nR}\}$.
	
The average probability of error for a code is defined as $P_e^{(n)} = P(M\neq \Psi(Y^n))$. A rate $R$ is said to be $(d,\infty)$-feedback achievable if there exists a sequence of $(n,2^{nR},(d,\infty))$ feedback codes, such that $\lim_{n\rightarrow \infty} P_e^{(n)} = 0$. The feedback capacity, $C_{(d,\infty)}^{\text{fb}}(\epsilon)$, is defined to be the supremum over all $(d,\infty)$-achievable rates, and is a function of the erasure probability, $\epsilon$.
\end{definition}

Our focus is on the binary erasure channel, or the BEC. Here, the input alphabet, $\mathscr{X} = \{0,1\}$, while the output alphabet is $\mathscr{Y} = \{0,?,1\}$, where $?$ denotes an erasure. Let $\epsilon \in [0,1]$ be the erasure probability of the channel.

\subsection{$Q$-graphs and $(S,Q)$-graphs}
\label{sec:Q-graph}
We now recall the definitions of the $Q$-graph and the $(S,Q)$-graph introduced in \cite{Single}.
\begin{definition}
	A $Q$-graph is a finite irreducible labelled directed graph on a vertex set $\mathscr{Q}$, with the property that each $q\in \mathscr{Q}$ has at most $|\mathscr{Y}|$ outgoing edges, each labelled by a unique $y\in \mathscr{Y}$.
\end{definition}

Thus, there exists a function $\Phi: \mathscr{Q} \times \mathscr{Y} \rightarrow \mathscr{Q}$, such that $\Phi(q,y)=q'$ if, and only if, there is an edge $q \stackrel{y}{\longrightarrow} q'$ in the $\mathscr{Q}$-graph.
Figure \ref{figQ} depicts a sample $\mathscr{Q}$-graph.

\begin{figure}[htbp]
	\centering
	\includegraphics[scale=0.6]{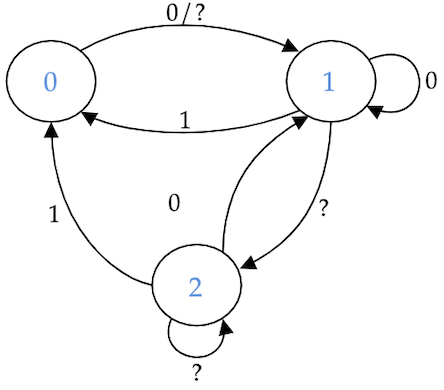}
	\caption{A sample $\mathscr{Q}$-graph.  The edge labels represent outputs.}
	\label{figQ}
\end{figure}
\begin{definition}
	Given an input-constrained DMC specified by $\{P(y|x)\}$ and the states of the presentation of the input constraint of which obey $s_t=f(s_{t-1},x_t)$, and a $\mathscr{Q}$-graph with vertex set $\mathscr{Q}$, the $(S,Q)$-graph is defined to be a directed graph on the vertex set $\mathscr{S}\times \mathscr{Q}$, with edges $(s,q) \xrightarrow{(x,y)} (s^{\prime},q^{\prime})$ if, and only if, $P(y|x)>0$, $s^{\prime}=f(s,x)$, and $q^{\prime}=\Phi(q,y)$.
	\label{def:SQ_graph}
\end{definition}

Given an input distribution $\{P(x|s,q)\}$ defined for each $(s,q)$ in the $(\mathscr{S},\mathscr{Q})$-graph, we have a Markov chain on $\mathscr{S}\times \mathscr{Q}$, where the transition probability associated with any edge $(x,y)$ emanating from $(s,q)\in \mathscr{S}\times \mathscr{Q}$ is $P(y|x)P(x|s,q)$. Let $\mathscr{G}(\{P(x|s,q)\})$ be the subgraph remaining after discarding edges of zero probability. We then define 

\begin{align*}
\Omega \triangleq \bigl\lbrace \{P(x | s,q)\}: \mathscr{G}& (\{P(x|s,q)\})\text{ has a single} \\ 
& \text{closed communicating class}\bigr\rbrace.
\end{align*}
An input distribution $\{P(x|s,q)\} \in \Omega$ is said to be \emph{aperiodic}, if the corresponding graph, $\mathscr{G}(\{P(x|s,q)\})$, is aperiodic. For such distributions, the Markov chain on $\mathscr{S}\times \mathscr{Q}$ has a unique stationary distribution $\pi(s,q)$.

\subsection{Bounds on Feedback Capacity}
We shall make use of the following single-letter upper bound on feedback capacity (specialized to input-constrained DMCs) \cite{Single}. The theorem assumes that the state transition graph corresponding to the input constraint is irreducible, and that the encoder and the decoder know the initial channel state, $s_0$.



\begin{theorem}[\cite{Single}, Theorem 2]\label{thm:UB}
	The feedback capacity, $C^{\text{fb}}_{DMC}$, of an input-constrained DMC, when the state transition graph of the input-constraint is irreducible, is upper bounded as
	\begin{equation*}
	C^{\text{fb}}_{DMC} \leq \sup_{P(x|s,q)\in \Omega} I(X;Y|Q),
	\end{equation*}
	for all irreducible $Q$-graphs with $q_0$ such that $(s_0,q_0)$ lies in an aperiodic closed communicating class.
\end{theorem}

\section{Main Results}
\label{sec:main}
\subsection{Capacity With Feedback}
The following theorem states our main result concerning the capacity of the $(d,\infty)$-RLL input-constrained BEC with feedback. For $\vec{\delta} = (\delta_0,\ldots,\delta_d)$, with $\delta_i\in [0,1], \ \forall i$, we define the function, $R(\vec{\delta})$, where $\vec{\delta}\in [0,1]^{d+1}$, to be
\begin{equation}
	\label{eq:R}
R(\vec{\delta}):= \frac{\bar{\epsilon}\left(\sum\limits_{i=0}^{d}\epsilon^ih_b(\delta_i)\right)}{\sum\limits_{i=0}^{d}\epsilon^i + d\bar{\epsilon}\left(\sum\limits_{i=0}^{d}\epsilon^i\delta_i\right)}.
\end{equation}
We let $\Delta_d$ denote the $d$-dimensional unit simplex, i.e., $\Delta_d:=\{\vec{\delta}\in [0,1]^{d+1}:\sum_{i=0}^d \delta_i\leq 1\}$.
\begin{theorem} \label{thm:main}
	For $\epsilon \in [0,1]$, the feedback capacity of the $(d,\infty)$-RLL input-constrained BEC is given by
	\begin{equation}
		\label{eq:fbcap}
	C_{(d,\infty)}^{\text{fb}}(\epsilon) = \max_{\vec{\delta}\in \Delta_d} R(\vec{\delta}),
	\end{equation}
	and is achievable by a zero-error feedback coding scheme.
\end{theorem}
\begin{remark}
	At $\epsilon=0$, the capacities with and without feedback are identical, and are given by $C_{(d,\infty)}(0) = \max\limits_{\delta \in [0,1]} \frac{h_b(\delta)}{d\delta + 1}$, the noiseless capacity of the $(d,\infty)$-RLL input constraint.
\end{remark}
\begin{remark}
	Since from operational considerations, the zero-error feedback capacity,  $C_{(d,\infty)}^{\text{ze}}$, is less than or equal to the feedback capacity, $C^{\text{fb}}_{(d,\infty)}$, Theorem \ref{thm:main} also shows that the two feedback capacities are indeed equal, i.e., $C^{\text{ze}}_{(d,\infty)}=C^{\text{fb}}_{(d,\infty)}$.
\end{remark}
%
Theorem \ref{thm:main} follows from the construction of a feedback coding scheme in Section \ref{sec:fbcode}, whose rate equals an upper bound on the feedback capacity computed using the single-letter bounding technique in \cite{Single}. The proof is presented in Section \ref{sec:fbcode}.

The expression for the feedback capacity provided in Theorem \ref{thm:main} admits the following simplification:

\begin{proposition}
	\label{prop:simplify}
The vector $\vec{\delta}^{\star}:=(\delta^*_0,\ldots,\delta^*_d)$ that attains the maximum in \eqref{eq:fbcap} is such that either $\vec{\delta}^{\star}$ is in the interior of $\Delta_d$, or that $\sum_{i=0}^{d} \delta_i = 1$. Further, in the first case, it holds that for any $\epsilon \in [0,1]$:
\[
C_{(d,\infty)}^{\text{fb}}(\epsilon) = \max_{\delta\in [0,\frac{1}{d+1}]} \frac{h_b(\delta)}{d\delta + \frac{1}{1-\epsilon}}.
\]
\end{proposition}
The proof of Proposition \ref{prop:simplify} is provided in Appendix A.

Figure \ref{fig:fbcap} shows plots of the feedback capacity for $d=1,2,3$. Several comments are now in order. The feedback capacity is equal to the noiseless capacity at $\epsilon=0$, and monotonically decreases to $0$ at $\epsilon=1$. As in \cite{0k}, numerical evaluations indicate that the feedback capacity is a concave function of the channel parameter, $\epsilon$. At $d=0$, which corresponds to the case of the BEC with no constraints, $R(0.5)$ equals $1-\epsilon$, which, in turn, equals the (feedback) capacity of the BEC with no input constraints. For $d=1$, it is easy to see that our feedback capacity expression recovers the formula for feedback capacity derived in \cite{SabBEC}. Indeed, for any $\epsilon$, it holds that:
\[
C^{\text{fb}}_{(1,\infty)}(\epsilon)\geq \max\limits_{\delta\in [0,\frac{1}{2}]}R(\delta,\delta) = \max\limits_{\delta\in [0,\frac{1}{2}]} \frac{h_b(\delta)}{\delta+\frac{1}{1-\epsilon}} = C^{\text{nc}}_{(1,\infty)}(\epsilon),
\]
where in the last equality, $C^{\text{nc}}_{(1,\infty)}(\epsilon)$ is the capacity with non-causal knowledge of erasures (or non-causal capacity), the expression for which was derived in \cite{SabBEC}. However, once again, it holds operationally that the feedback capacity is less than or equal to the non-causal capacity, thereby showing that $C^{\text{fb}}_{(1,\infty)} = C^{\text{nc}}_{(1,\infty)}(\epsilon)$, which agrees with the main result of \cite{SabBEC}. 

Similar reasoning leads us to the following corollary of Theorem \ref{thm:main} (stated as Corollary III.1 in \cite{arnk21}), for the case when $d=2$:
\begin{corollary}
	For $d=2$ and $\epsilon\in [0,1-\frac{1}{2\log_2(3/2)}]$, it holds that the feedback and non-causal capacities are equal, i.e.,
	\[
	C^{\text{fb}}_{(2,\infty)}(\epsilon) = C^{\text{nc}}_{(2,\infty)}(\epsilon) = \max\limits_{\delta\in [0,\frac{1}{3}]} \frac{h_b(\delta)}{2\delta+\frac{1}{1-\epsilon}}.
	\]
\end{corollary}
\begin{proof}
	We note, as before, that it suffices to show that for $\epsilon\in [0,1-\frac{1}{2\log_2(3/2)}]$, it holds that $C^{\text{fb}}_{(2,\infty)}(\epsilon)\geq C^{\text{nc}}_{(2,\infty)}(\epsilon)$.
	
	From \cite{0k}, we know that the non-causal capacity of the $(d,\infty)$-RLL input-constrained BEC is given by:
	\[
	C^{\text{nc}}_{(d,\infty)}(\epsilon) = \max\limits_{\delta\in [0,\frac{1}{2}]} V(\delta),
	\]
	where $V(\delta) = \frac{h_b(\delta)}{d\delta+\frac{1}{1-\epsilon}}$. We note that the derivative, $V^{\prime}(\cdot)$, is given by
	\begin{equation*}
		V^{\prime}(\delta) = \frac{(\kappa+d)\log(1-\delta)-\kappa\log\delta}{(\kappa+d\delta)^2},
	\end{equation*} 
	where we write $\frac{1}{1-\epsilon}$ as $\kappa$. It can be checked that $V^\prime(\delta)$ is strictly decreasing in $\delta$.
	 Now, when $d=2$, it is true that $V^\prime\left(\frac{1}{3}\right)\leq 0$ if, and only if, $\kappa \leq \frac{2 \log(1+1/2)}{\log 2}$, or, equivalently, if, and only if, $\epsilon \leq 1-\frac{1}{2\log(\frac{3}{2})}$. Since $V^\prime(0^+)>0$, we have that for $\epsilon \leq 1-\frac{1}{2\log(\frac{3}{2})}$, the unique maximum of $V(\cdot)$, over $[0,1]$, occurs in the interval $[0,\frac{1}{3}]$. Hence, for $\epsilon \in [0,1-\frac{1}{2\log_2(3/2)}]$, we have that
	 \[
	 C^{\text{nc}}_{(d,\infty)}(\epsilon) = \max\limits_{\delta\in [0,\frac{1}{3}]} V(\delta).
	 \]
	Furthermore, from Theorem \ref{thm:main}, we get that
	\[
	C^{\text{fb}}_{(2,\infty)}(\epsilon)\geq \max\limits_{\delta\in [0,\frac{1}{3}]}R(\delta,\delta,\delta) = \max\limits_{\delta\in [0,\frac{1}{3}]} V(\delta) = C^{\text{nc}}_{(2,\infty)}(\epsilon).
	\]
\end{proof}
We note from the observations in \cite{0k} that this equality of feedback and non-causal capacities is not true for general $d$.

Figures \ref{fig:thang1} and \ref{fig:thang2} show comparisons of the feedback capacities for $d=1$ and $d=2$, respectively, with dual capacity-based upper bounds on the capacities without feedback, derived in \cite{Thangaraj}. Clearly, it holds that the capacity without feedback is less than the feedback capacity, for this class of channels. 

\begin{figure*}[t]
	\centering
	\includegraphics[scale=0.13]{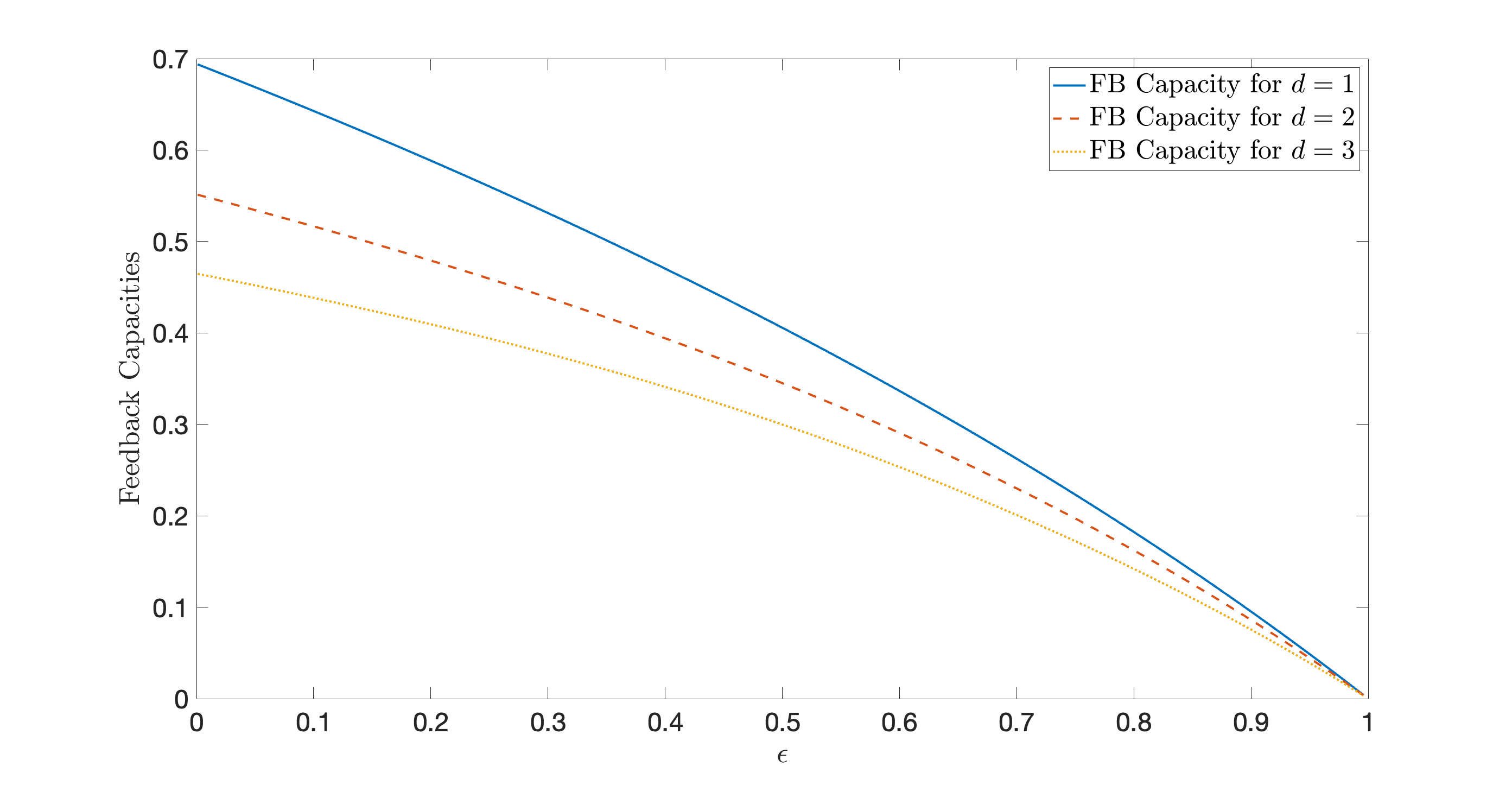}
	\caption{Plots of the feedback capacities for $d=1,2,3$.}
	\label{fig:fbcap}
\end{figure*}
%

\begin{figure*}[t]
	\centering
	\subfloat[]{\includegraphics[width=0.5\linewidth]{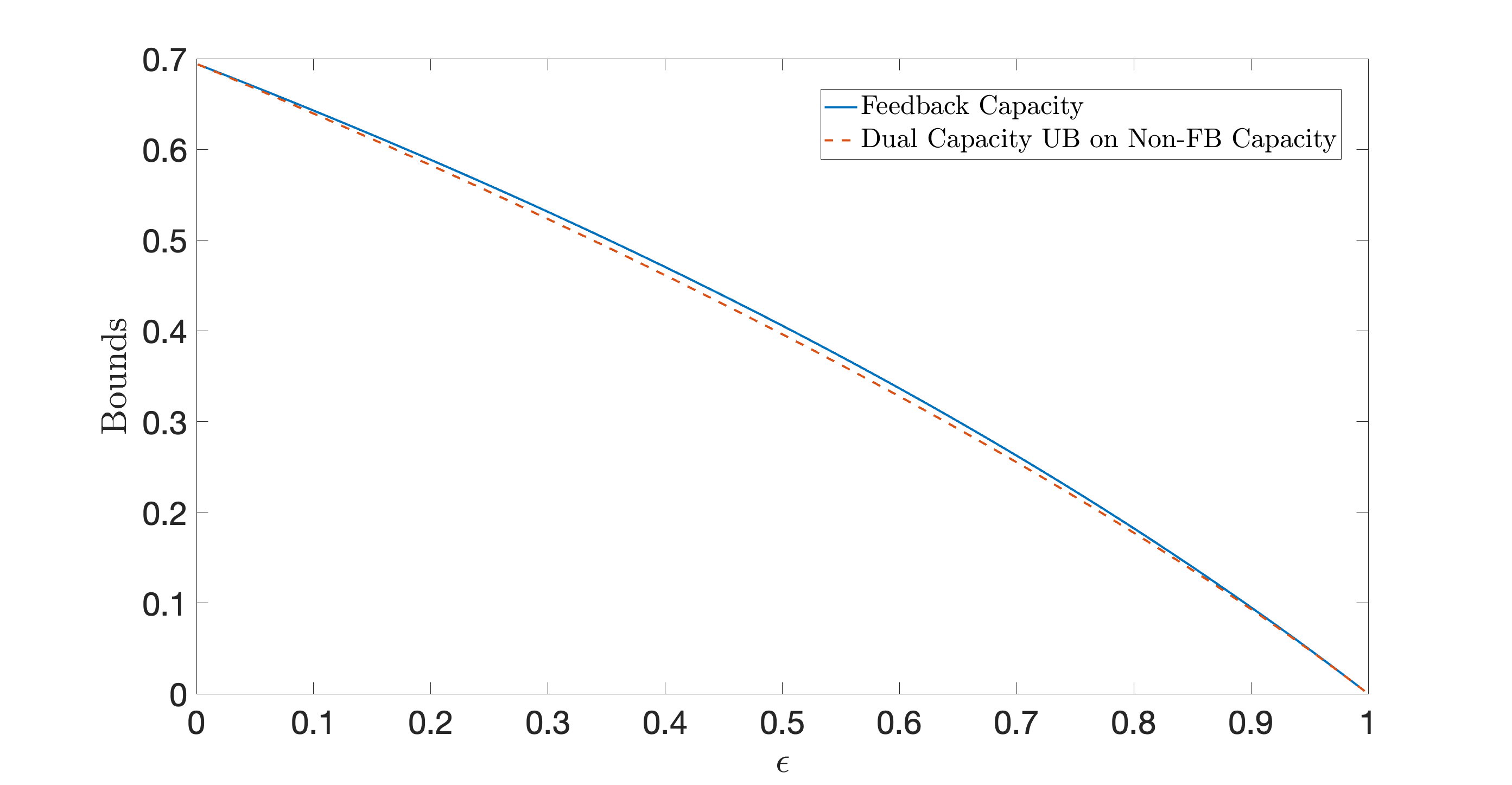}\label{fig:thang1}}
	\hfill
	\subfloat[]{\includegraphics[width=0.5\linewidth]{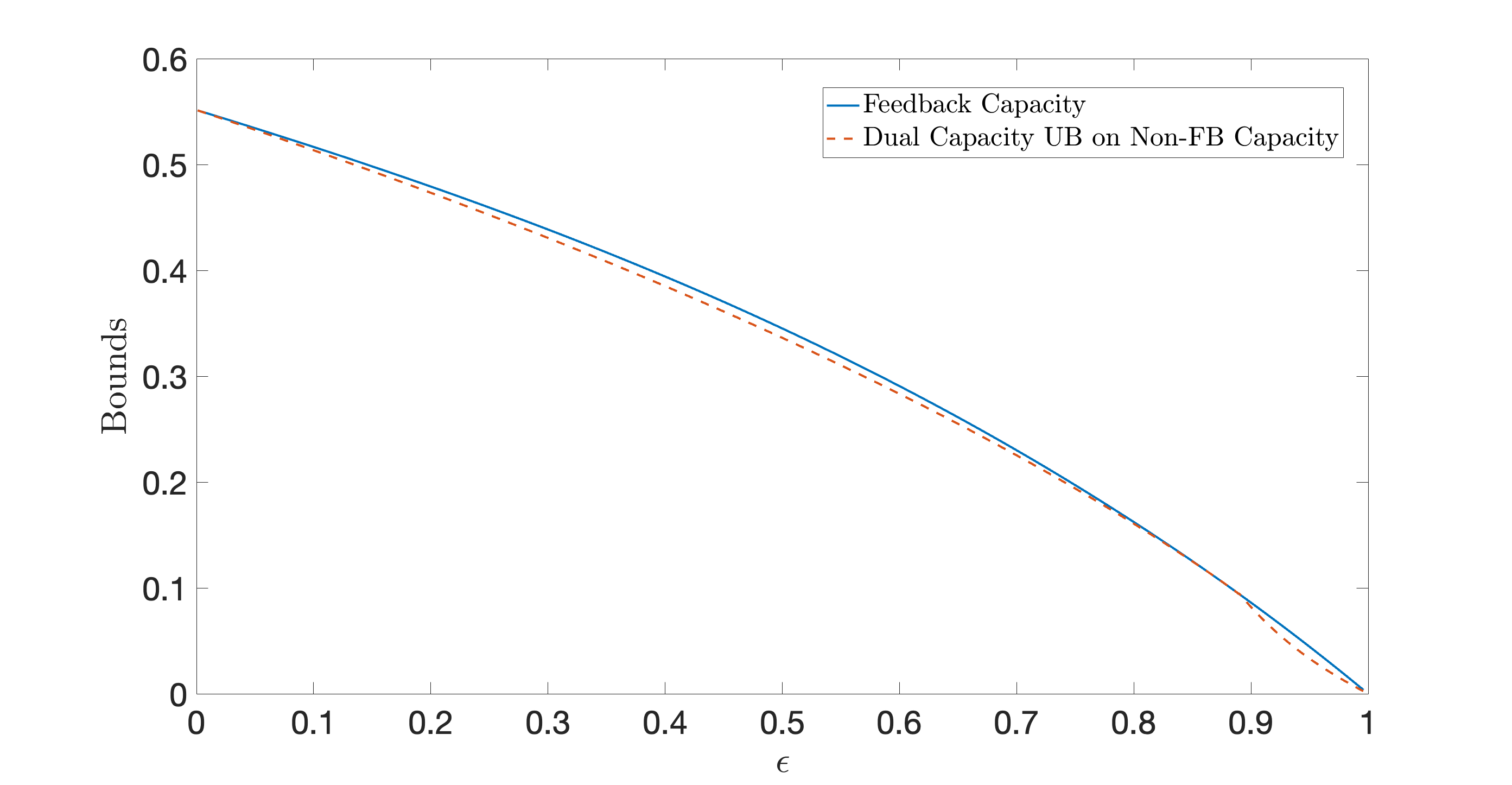}\label{fig:thang2}}
	\caption{Plots (a) and (b) show comparisons of the feedback capacities of the $(1,\infty)$- and $(2,\infty)$-RLL input-constrained BEC with dual capacity-based upper bounds on the capacity without feedback, from \cite{Thangaraj}.}
\end{figure*}

\section{Optimal Feedback Coding Scheme}
\label{sec:fbcode}
This section presents a simple feedback coding scheme that achieves the lower bound in Theorem \ref{thm:main}. Our labelling-based coding scheme is similar to the coding schemes in \cite{SabBEC} and \cite{0k}. The main feature of the scheme is a dynamically-changing set of possible messages that is known to both the encoder and the decoder at all times. The objective of the encoder is to communicate a sequence of bits that will enable the decoder to narrow down the set of possible messages to a single message.

Each message $m\in [2^{nR}]$ is mapped uniformly to a point in the unit interval, i.e., the message $m$ is mapped to the point $\frac{m-1}{2^{nR}}$. At each time instant $i$, the unit interval is partitioned into sub-intervals that are labelled by either a `$0$' or a `$1$'. The input $x_i$ to the channel is determined using the label of the sub-interval containing the message. 

The coding scheme proceeds as follows: we first pick positive real numbers $\delta_0,\delta_1,\ldots,\delta_d$ such that $\sum\delta_i\leq 1$. To determine the input bits to be sent, we use a set of $d+2$ labellings, $\mathcal{L}_0,\ldots, \mathcal{L}_d,\hat{\mathcal{L}}$, with the interval $[\sum_{j<i}\delta_j,\sum_{j\leq i}\delta_j)$, in $\mathcal{L}_i$, labelled by a `$1$'. Further, labelling $\hat{\mathcal{L}}$ is such that the entire interval $[0,1)$ is labelled by a `$0$'. Figure \ref{fig:labellings} shows an illustration of the set of labellings.

\begin{figure}[htbp]
	\centering
	\includegraphics[width=0.4\textwidth]{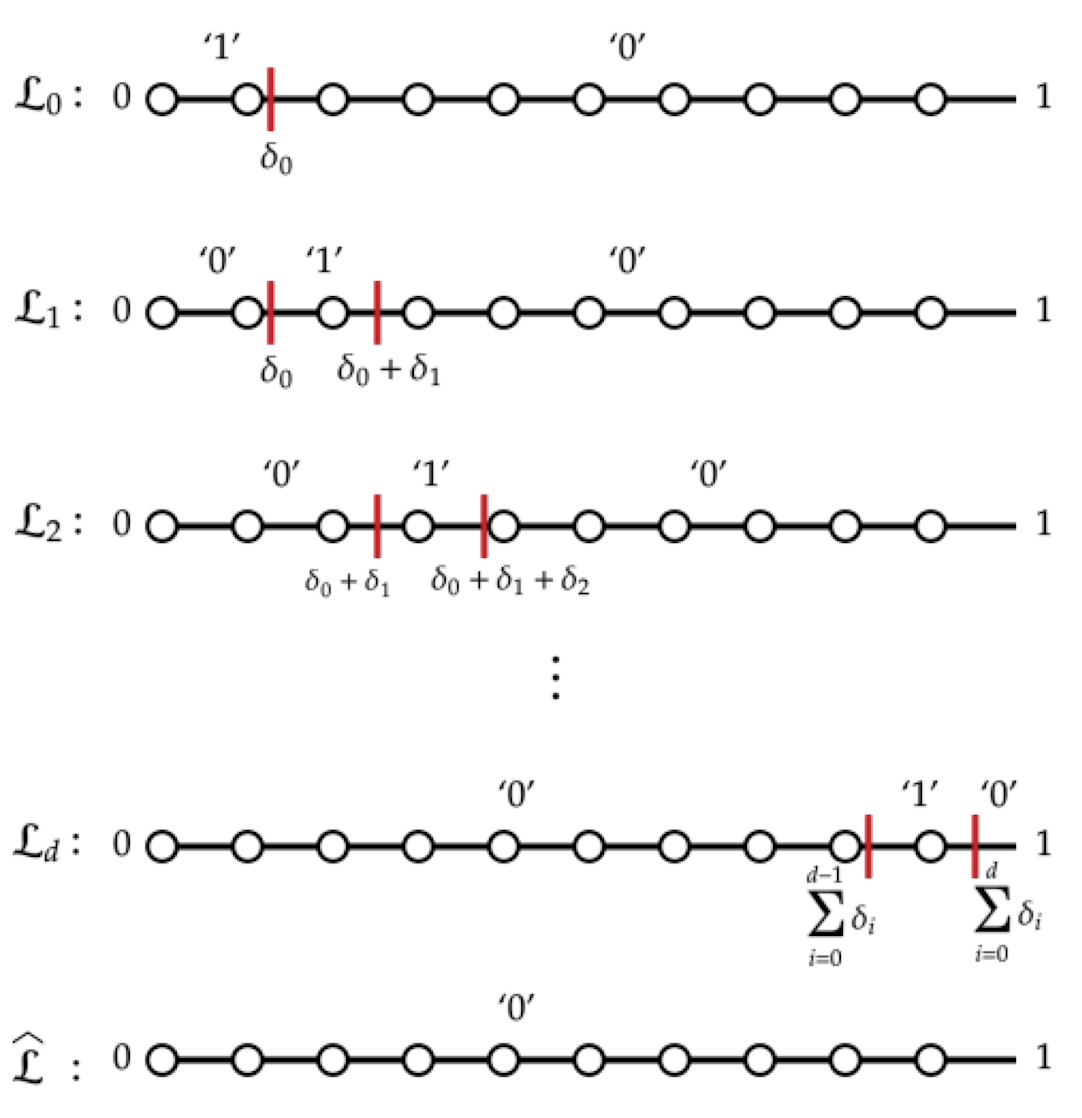}
	\caption{The set of labellings, $\mathcal{L}_0,\ldots, \mathcal{L}_d$, used in the coding scheme.}
	\label{fig:labellings}
\end{figure}

\begin{algorithm}[t]
	\caption{Coding Scheme}
	\label{alg}
	\begin{algorithmic}[1]		
		\Procedure{Code}{m}       
		\State Set $\mathcal{M}_0 = [2^{nR}]$ and $Y_{-d}^{0} = (0,\ldots,0)$.
		\State Set Label = $\mathcal{L}_0$.
		\State Set time index $i = 1$.
		\While{$|\mathcal{M}_{i-1}|>1$}
		\State Transmit $X_i = \Gamma(m,\text{Label})$. \Comment{Encoder}
		\If{none of $Y_{i-d-1}^{i-1}$ is a $1$}
		\If{$Y_i = 0$}
		\State Set $\mathcal{M}_i = \mathcal{M}_{i-1}^{(0)}$.
		\ElsIf{$Y_i=1$}
		\State Set $\mathcal{M}_i =\mathcal{M}_{i-1}^{(1)}$. 
		\Else
		\State Set $\mathcal{M}_i =\mathcal{M}_{i-1}$. 
		\EndIf
		\Else
		\State Set $\mathcal{M}_i =\mathcal{M}_{i-1}$. 
		\EndIf
		\State Label $\leftarrow$ $G(\text{Label},Y_{i-d}^{i})$.
		\State Update $i = i+1$.
		\EndWhile
		\State Output $\hat{m}$, where $\mathcal{M}_{i-1} = \{\hat{m}\}$. \Comment{Decoder}
		\EndProcedure
		
	\end{algorithmic}
\end{algorithm} 
The {labelling to be used at any time instant} is a function of the channel outputs upto that time instant and can be recursively computed using the previous labelling and previous $d+1$ channel outputs. Thus, both the encoder and the decoder can compute the labelling used at all times. Let $L_i$ denote the labelling used at time $i$. We fix $L_0 := \mathcal{L}_0$. Then, $L_{i+1} = G(L_i,Y^{i}_{i-d})$, where the function $G$ is defined as follows:

\[G(L_i,Y^{i}_{i-d+1}) =
\begin{cases}
	\hat{\mathcal{L}},\ \text{ if one of $Y_{i-d+1}^{i}$ is a 1},\\
	\mathcal{L}_0,\ \text{if $Y_{i-d} =1$,}\\
	\mathcal{L}_0,\ \text{if $Y_i=0$},\text{ and none} \text{ of $Y_{i-d+1}^{i}$ is a 1},\\
	\mathcal{L}_{\text{$j+1$ mod$(d+1)$}},\ \text{if $L_i = \mathcal{L}_j$, $Y_{i}=?$ and none} \\\ \ \ \ \ \ \ \ \ \ \ \ \ \ \ \ \  \text{of $Y_{i-d+1}^{i}$ is a 1}.\\
\end{cases}
\]

The transitions between the labellings can be represented by the finite state machine (FSM) shown in Figure \ref{fig:QUB}. The labellings in conjunction with the message are used to determine the bit $X_{i+1}$ to be transmitted. Formally, $X_{i+1} = \Gamma(m,L_{i+1})$, where the function $\Gamma$ is defined as:

\[\Gamma(m,L_{i+1}) =
\begin{cases}
	0,\ \text{if $L_{i+1} = \mathcal{L}_j$, for some $j$, and}\\ \ \ \ \ \ \ \ \ \ \ \ \ \ \text{$\frac{m-1}{2^{nR}} \notin [\sum_{z<j}\delta_z,\sum_{z\leq j}\delta_z)$},\\
	0,\ \text{if $L_{i+1} = \hat{\mathcal{L}}$},\\
	1,\ \text{if $L_{i+1} = \mathcal{L}_j$, for some $j$, and}\\ \ \ \ \ \ \ \ \ \ \ \ \ \ \text{$\frac{m-1}{2^{nR}} \in [\sum_{z<j}\delta_z,\sum_{z\leq j}\delta_z)$}.
\end{cases}
\]
The chronological order of a single use of the channel at time $i$, for a fixed message $m$, thus is $L_i\rightarrow X_i\rightarrow Y_i$, with $L_{i+1} = G(L_i,Y_{i-d}^{i})$.

For a given output sequence $y^i$, we denote by $\mathcal{M}_i$ the set of \emph{possible messages} after time $i$, i.e., $\mathcal{M}_i := \{m\in [2^{nR}]: P(m|y^{i})>0\}$, with $\mathcal{M}_0:=[2^{nR}]$. Note that both the encoder and the decoder can compute the conditional distribution $P(m|y^{i})$ using Bayes' rule. We use the notation $\mathcal{M}_i^{(0)}$ and $\mathcal{M}_i^{(1)}$ to denote the set of messages labelled by a `$0$' and by a `$1$', respectively, in $\mathcal{M}_i$.

A transmission at time $i$ is said to be \emph{successful} if $|\mathcal{M}_i|<|\mathcal{M}_{i-1}|$. Specifically, a successful transmission can occur in one of two scenarios: the first is $y_i = 1$, and the second is where $y_i = 0$ and none of $y_{i-d}^{i-1}$ is a $1$. After a successful transmission, the set of possible messages is calculated and expanded uniformly to the unit interval. It is easy to see that in the first kind of successful transmission, the new set of possible messages, $\mathcal{M}_i$, is equal to $\mathcal{M}_{i-1}^{(1)}$, and in the second kind, $\mathcal{M}_i$ equals $\mathcal{M}_{i-1}^{(0)}$. Figure \ref{fig:succ} shows an illustration of the second of the two kinds of successful transmissions. Figure \ref{fig:eras} depicts the situation when a sequence of erasures is received. In this case, the new set of possible messages, $\mathcal{M}_i$, is equal to $\mathcal{M}_{i-1}$. This transmission procedure continues repeatedly until the set of possible messages contains one message. Additionally, we note that the coding scheme is zero-error, since at the end of the algorithm, the decoder can uniquely decode the transmitted message. The encoding and decoding procedures are described in Algorithm \ref{alg}. 
\begin{figure*}[!t]
	\centering
	
	\includegraphics[width=0.45\textwidth]{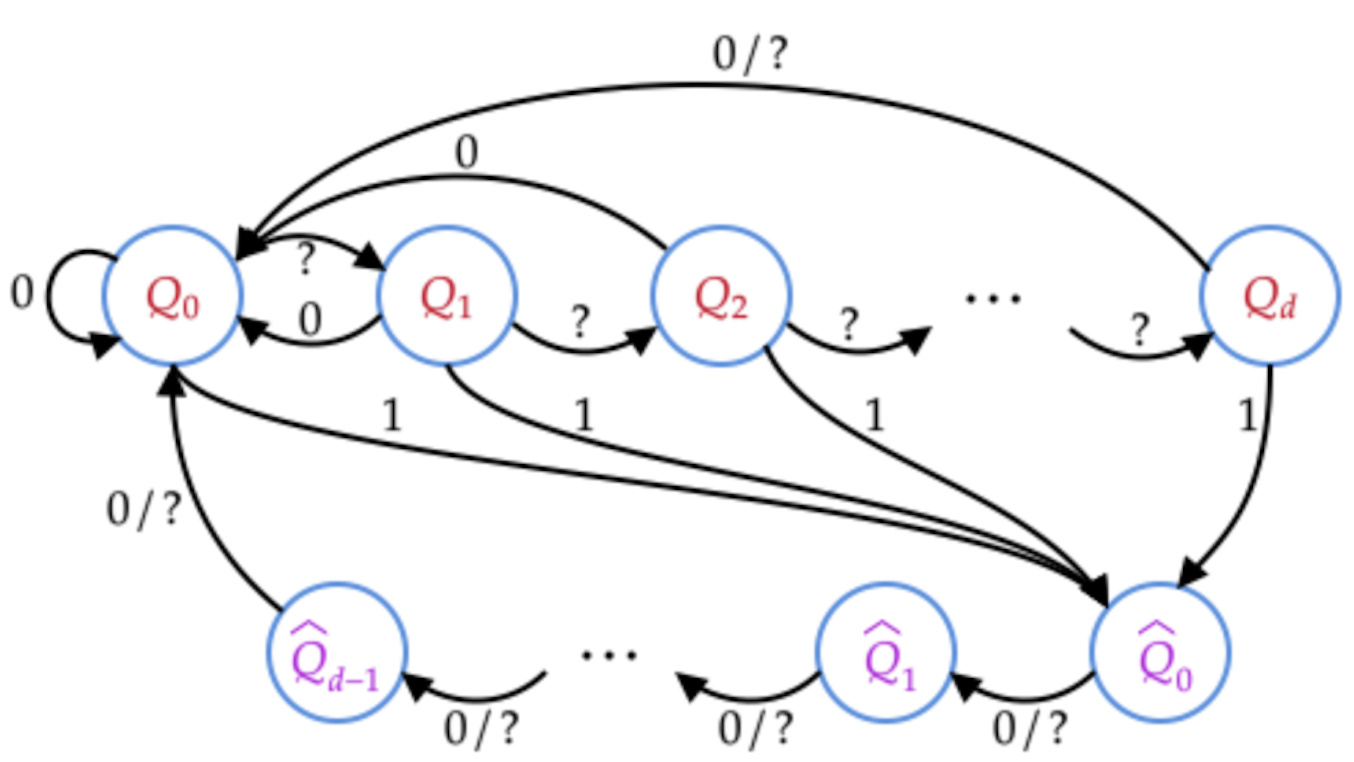}
	\caption{Figure shows the finite-state machine (FSM) that represents transitions between the labellings, with the edges labelled by outputs. When the encoder is in state $Q_i$, for $i\in \{0,1,\ldots,d\}$, the labelling used is $\mathcal{L}_i$, and when the encoder is in state $\hat{Q}_i$, for $i\in \{0,1,\ldots,d-1\}$, the labelling used is $\hat{\mathcal{L}}$.}
	\label{fig:QUB}
\end{figure*}

The construction of the labellings ensures that the $(d,\infty)$-RLL input constraint is obeyed. This can be seen from the fact that if ever $Y_i = 1$ for some $i$, it holds that the next $d$ inputs are all set to be $0$s. Further, the staggered manner in which the sub-intervals are labelled by a $1$ in the labellings $\mathcal{L}_0,\ldots,\mathcal{L}_d$ ensures that the input constraint is satisfied when a sequence of erasures is received, too. 

The analysis of the rate of the feedback coding scheme is similar to the proofs of Lemma 3 and Lemma 4 in \cite{0k}. In order to make the exposition self-contained, we repeat parts of the proofs here. For a time index $i\in [n]$, we define $\mathsf{J}_i$ to be the number of information bits gained in a single channel use, which is the logarithm of the change in the size of possible messages at the end of the channel use, i.e.,
\begin{equation}
\label{eq:Jk}
\mathsf{J}_i:= \log_2|\mathcal{M}_{i-1}|-\log_2|\mathcal{M}_{i}|.
\end{equation}
 and $\mathsf{L}_i$ to denote the random variable corresponding to the labelling used, with $\mathsf{L}_i\in \mathcal{L} := \{\hat{\mathcal{L}},\mathcal{L}_0,\ldots,\mathcal{L}_d\}$. The following lemma then holds true:

\begin{lemma}
	\label{lem:gain}
	For any time $i\in [n]$, and for all $\ell\in \mathcal{L}$, we have that $\mathbb{E}[\mathsf{J}_i| \mathsf{L}_i=\ell] = \bar{\epsilon}h_b(\delta_\ell)$, where
	\[\delta_\ell = 
	\begin{cases}
		0,\ \text{if $\ell = \hat{\mathcal{L}}$},\\
		\delta_j,\ \text{if $\ell = \mathcal{L}_j$}.
	\end{cases}
	\]
\end{lemma}
\begin{figure}[!t]
	\centering
	
	\includegraphics[width=0.5\textwidth]{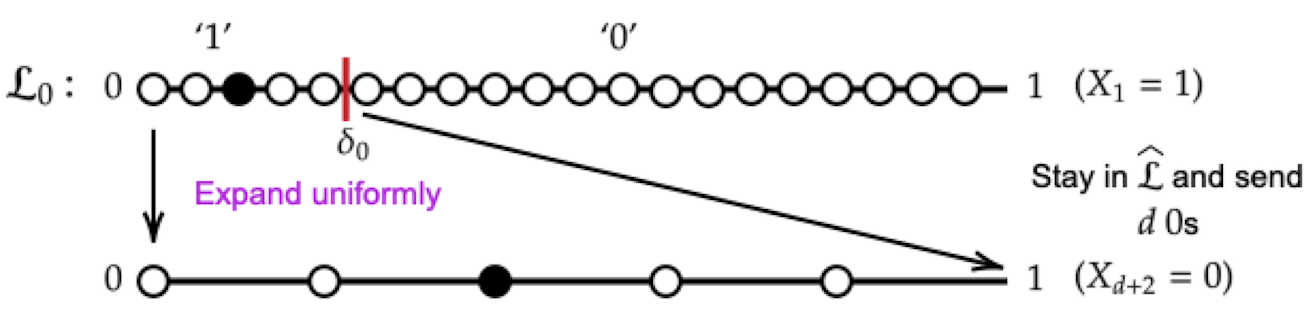}\label{<figure2>}
	\caption{Figure shows an illustration of a successful transmission of the second kind, when $X_1 = Y_1 = 1$, and when the encoder then transmits $d$ additional zeros. Note that the size of the set of possible messages reduces.}
	\label{fig:succ}
\end{figure}

\begin{figure}[!t]
	\centering
	\includegraphics[width=0.46\textwidth]{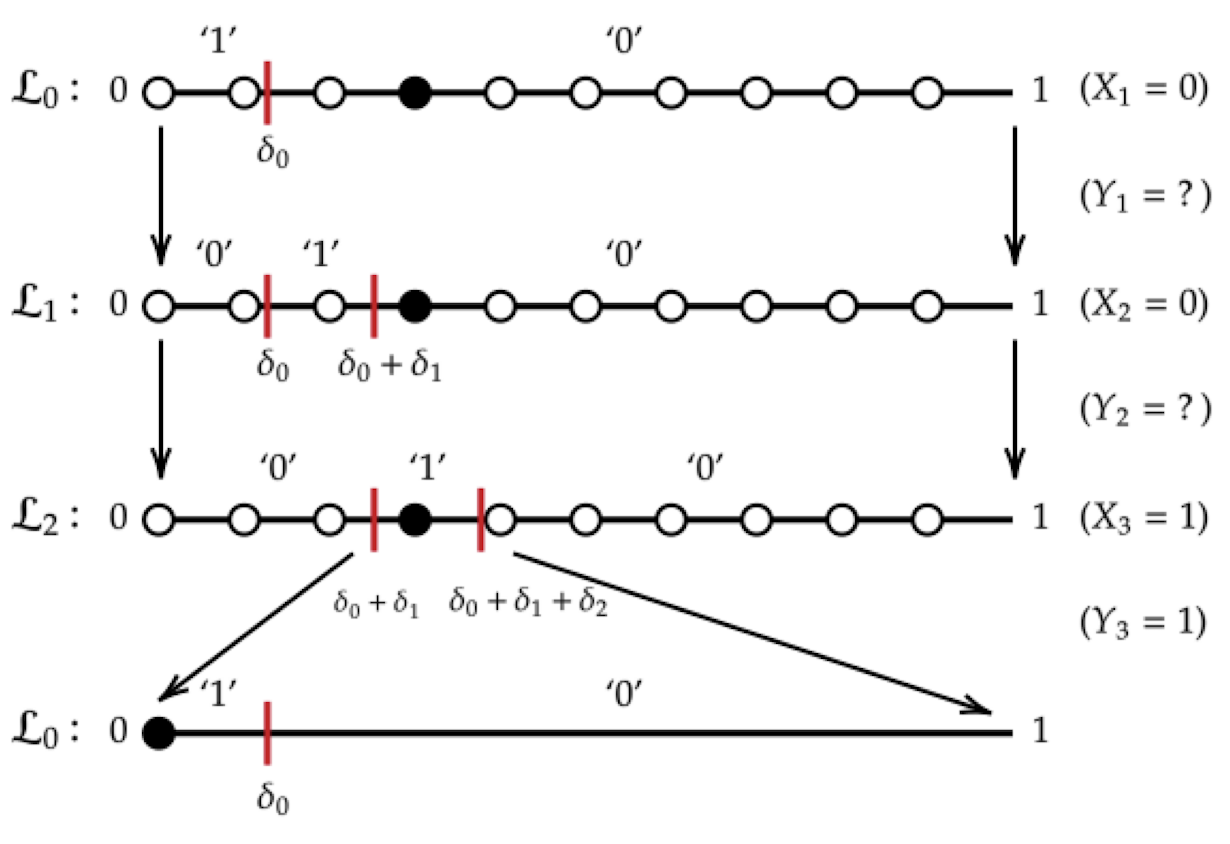}
	\caption{The figure shows the setting when two consecutive erasures are received ($Y_1=Y_2=?$), followed by the successful reception of $X_3 = 1$. So long as erasures are received, the set of possible messages is retained as such, while the labellings cycle through $\mathcal{L}_0$ to $\mathcal{L}_d$. Upon the successful reception of $X_3$, and after the transmission of $d$ $0$s, the labelling is changed to $\mathcal{L}_0$. However, since the set of possible messages is now a singleton, the transmission ends with the decoder declaring the correct identity of the message.}
	\label{fig:eras}
\end{figure}

\begin{proof}
	We let the random variable $\theta$ denote the indicator that the current output is unerased, i.e.,
	\[
	\theta_i = \mathds{1}\{Y_i\neq `?'\}.
	\]
	Then,
	\begin{align}
		\mathbb{E}[\mathsf{J}_i|\mathsf{L}_i=\ell]&=\epsilon\mathbb{E}[\mathsf{J}_i|\mathsf{L}_i=\ell,\theta_i = 0]+\bar{\epsilon}\mathbb{E}[\mathsf{J}_i|\mathsf{L}_i=\ell,\theta_i = 1] \notag\\
		&= \bar{\epsilon}\mathbb{E}[\mathsf{J}_i|\mathsf{L}_i=\ell,\theta_i = 1] \label{eq:lemm},
	\end{align}
where the second equality holds since if an erasure is received, the set of possible messages remains the same.

We now note that in any of the labellings $\mathcal{L}_j$, it holds that the length of the sub-interval labelled by a `$1$' equals $\delta_j$, and equals $0$ in $\hat{\mathcal{L}}$. Hence, if labelling $\ell \in \mathcal{L}$ is employed, the bit transmitted is distributed according to Ber$(\delta_\ell)$ (see the remark immediately after this proof), where $\delta_\ell$ is as defined in the statement of the theorem.

 Assume that the current size of the set of possible messages, $\mathcal{M}$, equals $k$, with the current labelling used being $\ell$. If the current bit transmitted is received successfully, then, the new set of possible messages has size $k\delta_\ell$, if the current input bit is a `$1$', and is equal to $k\bar{\delta}_\ell$, otherwise. Hence, the expected number of bits required to describe the new set of possible messages is $\bar{\delta}_{\ell} k \log_2(\bar{\delta}_{\ell} k)+\delta_\ell k \log_2(\delta_\ell k) = \log_2(k)-h_b(\delta_{\ell})$. Thus, given that $\mathsf{L}=\ell$, following a successful transmission, the decoder gains $h_b(\delta_{\ell})$ bits of information. Substituting into \eqref{eq:lemm}, we get that
\[
\mathbb{E}[\mathsf{J}|\mathsf{L}=\ell] = \bar{\epsilon}h_b(\delta_\ell).
\]
\end{proof}

\begin{remark}
	\label{rem:ber}
	We note as in \cite{0k} that since the messages are discrete points in $[0,1)$, the transmitted bit is actually distributed according to Ber$(\delta_\ell -e_i)$, where $e_i$ is a correction factor that is bounded as $0\leq e_i\leq \frac{1}{|\mathcal{M}_{i-1}|}$. We use Algorithm \ref{alg} for encoding, until a time $t$ such that $|\mathcal{M}_t|\leq 2^{\lambda}$, for an absolute constant, $\lambda>0$. A clean-up coding phase can then be employed, after the labelling-based scheme, similar to \cite[Appendix C]{SabBEC}, using which the remaining at most $\lambda$ bits are then transmitted. The rate of the overall two-stage coding scheme can then be made arbitrarily close to the rate calculated using the analysis in this paper.
\end{remark}

The following lemma computes the rate of the proposed coding scheme:
\begin{lemma}
	\label{lem:rate}
	For any $\epsilon \in [0,1]$, the proposed coding scheme achieves a rate
	\[
	R = \max_{\vec{\delta}\in \Delta_d} R(\vec{\delta}),
	\]
	where $R(\vec{\delta})$ is as defined in equation \eqref{eq:R}.
\end{lemma}
\begin{proof}
	Fix $0\leq \delta_0,\delta_1,\ldots,\delta_d\leq 1$, with $\sum_{i=0}^{d}\delta_i \leq 1$. The constraint is chosen so as to ensure that the lengths of the intervals used in the labellings in Figure \ref{fig:labellings} are all non-negative. 
	
	Consider the Markov chain induced by the transitions shown in Figure \ref{fig:QUB}, with the transition probabilities $P(Y=?|q) = \epsilon$, for all $q \in \{Q_0,Q_1,\ldots,Q_d,\hat{Q}_0,\hat{Q}_1,\ldots,\hat{Q}_{d-1}\}$, and
	\[
	P(Y=1|q) = 
	\begin{cases}
		\bar{\epsilon}\delta_i, \ \text{if $q = Q_i$, for some $i$},\\
		0, \ \text{otherwise}.
	\end{cases}
	\]
	Let $\pi(\ell)$ denote the stationary probability of using labelling $\ell$, which can be calculated using the stationary probabilities of states of the Markov chain on the FSM. Clearly, for $i\in \{0,1,\ldots,d\}$, $\pi(\mathcal{L}_i)$ equals $\pi(Q_i)$, and $\pi(\hat{\mathcal{L}})$ equals $\sum_{j=0}^{d-1}\pi(\hat{Q}_j)$. The rate of the coding scheme can be computed to be:
	\begin{align*}
		R &= \lim_{n\rightarrow \infty}  \frac{\log_2 |\mathcal{M}_0|}{n}\\
		&\stackrel{(a)}{=} \lim_{n\rightarrow \infty} \frac{1}{n}\sum_{k=1}^{n}\mathbb{E}[\mathsf{J}_k]\\
		&= \lim_{n\rightarrow \infty}  \frac{1}{n} \sum_{k=1}^{n} \sum_{\ell \in \mathcal{L}} P(\mathsf{L}_k = \ell)\mathbb{E}[\mathsf{J}_k|\mathsf{L}_k = \ell]\\
		&\stackrel{(b)}{=} \sum_{\ell \in \mathcal{L}} \bar{\epsilon}h_b(\delta_\ell) \lim_{n\rightarrow \infty}  \frac{1}{n} \sum_{k=1}^{n} P(\mathsf{L}_k = \ell)\\
		&\stackrel{(c)}{=} \sum_{\ell \in \mathcal{L}} \bar{\epsilon}h_b(\delta_\ell)\pi(\ell)\\
		&= \sum_{i=0}^{d}\bar{\epsilon}h_b(\delta_i)\pi(\mathcal{L}_i)\\
		&\stackrel{(d)}{=} R(\vec{\delta}),
	\end{align*}
where 
\begin{enumerate}[label = (\alph*)]
	\item holds from equation \eqref{eq:Jk},
	\item follows from Lemma \ref{lem:gain} and by exchanging the order of the summations,
	\item follows from the definition of stationary probability and from the fact that the random process $(\mathsf{L}_n:n\in \mathbb{N})$ is a positive recurrent, irreducible, aperiodic Markov chain (this follows from the irreducibility and aperiodicity of the graph in Figure \ref{fig:QUB}), and
	\item holds by an explicit calculation of the stationary probabilities of the Markov process $(\mathsf{L}_n:n\in \mathbb{N})$, which are given by:
	\[
	\pi(\ell)=
	\begin{cases}
		\frac{\epsilon^i}{\sum_{j=0}^{d}\epsilon^j + d\bar{\epsilon}\left(\sum_{j=0}^{d}\epsilon^j\delta_j\right)},\ \ell=\mathcal{L}_i,\\
		\frac{d\bar{\epsilon}\left(\sum_{j=0}^{d}\epsilon^j\delta_j\right)}{\sum_{j=0}^{d}\epsilon^j + d\bar{\epsilon}\left(\sum_{j=0}^{d}\epsilon^j\delta_j\right)},\ \ell=\hat{\mathcal{L}}.
	\end{cases}
	\]
\end{enumerate}
Therefore, maximizing over all allowed parameters $\delta_0,\ldots,\delta_d$, it follows that our coding scheme achieves a rate $R = \max_{\vec{\delta}\in \Delta_d} R(\vec{\delta})$.
\end{proof}

We now provide an upper bound on the feedback capacity of the $(d,\infty)$-RLL input-constrained BEC, using Theorem \ref{thm:UB}.
\begin{lemma}
	\label{lem:ub}
	For any $\epsilon \in [0,1]$, the feedback capacity of the $(d,\infty)$-RLL input-constrained BEC is bounded as:
	\[
	C^{\text{fb}}_{(d,\infty)}\leq \max_{\vec{\delta}\in \Delta_d} R(\vec{\delta}).
	\]
\end{lemma}
\begin{proof}
	We simply apply Theorem \ref{thm:UB} to the $Q$-graph that is the FSM in Figure \ref{fig:QUB}. The transition probabilities $P(Y=y|q)$ are the same as those in the proof of Lemma \ref{lem:rate}, for $y\in \{0,?,1\}$ and $q \in \{Q_0,Q_1,\ldots,Q_d,\hat{Q}_0,\hat{Q}_1,\ldots,\hat{Q}_{d-1}\}$. Note that the definition of the transition probabilities implies that 
	\[
	P(X=1|q) = 
	\begin{cases}
		\delta_i, \ \text{if $q = Q_i$, for some $i$},\\
		0, \ \text{otherwise}.
	\end{cases}
	\]
	 We then have that
	 \begin{align*}
	 	I(X;Y|Q) &\stackrel{(a)}{=} H(Y|Q) - H(Y|X)\\
	 	&= H(Y|Q) - h_b(\epsilon)\\
	 	&\stackrel{(b)}{=} \bar{\epsilon}H(X|Q)+h_b(\epsilon) - h_b(\epsilon)\\
	 	&= \sum_{\ell \in \mathcal{L}} \bar{\epsilon}h_b(\delta_\ell)\pi(\ell)\\
	 	&= \sum_{i=0}^{d}\bar{\epsilon}h_b(\delta_i)\pi(\mathcal{L}_i)\\
	 	&= R(\vec{\delta}),
	 \end{align*}
 where
 \begin{enumerate}[label = (\alph*)]
 	\item holds due to the memorylessness of the BEC, and
 	\item follows from the simple identity that $H(a\bar{c},\bar{a}\bar{c},c) = h_b(c)+\bar{c}h_b(a)$, for all $a,c \in [0,1]$.
 \end{enumerate}
Hence, all that remains to be shown for the proof to be complete is that 
\[
\sup_{\{P(x|s,q)\}} R(\vec{\delta}) = \max_{\sum_{i=0}^{d}\delta_i \leq 1}R(\vec{\delta}).
\]
This is true since for any valid input distribution $\{P(x|s,q)\}$ on the $(S,Q)$-graph corresponding to the $Q$-graph in Figure \ref{fig:QUB}, it holds that
\begin{align*}
	\sum_{i=0}^{d}\delta_i = \sum_{i=0}^{d}P(X=1|Q_i)
	&= \sum_{i=0}^{d}\frac{P(X=1,Q_i)}{\pi(\mathcal{L}_i)}\\
	&{=} \sum_{i=0}^{d}\frac{P(X=1,Q_i)}{\epsilon^i \pi(\mathcal{L}_0)}\\
	&= \frac{1}{ \pi(\mathcal{L}_d)}\sum_{i=0}^{d}\epsilon^{d-i}P(X=1,Q_i).	
\end{align*}
Now,
\begin{align*}
    \sum_{i=0}^{d}\delta_i &= \frac{1}{ \pi(\mathcal{L}_d)}\sum_{i=0}^{d}\epsilon^{d-i}P(X=1,Q_i)\\
    &= \frac{1}{ \pi(\mathcal{L}_d)}\sum_{i=0}^{d}\epsilon^{d-i}P(X=1,S=d,Q_i)\\
    &\leq \frac{1}{ \pi(\mathcal{L}_d)}\left(\sum_{i=0}^{d-1}\epsilon^{d-i}P(X=1,S=d,Q_i)+P(S=d,Q_d)\right)\\
    &= \frac{1}{ \pi(\mathcal{L}_d)}\left(\sum_{i=0}^{d-1}P(S=i,Q_d)+P(S=d,Q_d)\right)\\
    &=1,
\end{align*}
where the penultimate equality holds from the structure of the $(S,Q)$-graph.
\end{proof}
\begin{remark}
	We note that for the case when $d=2$, the authors in \cite{0k} provide the optimal $Q$-graph of Figure \ref{fig:QUB} for evaluating an upper bound on the feedback capacity. We have shown here that their upper bound is tight.
\end{remark}
Lemmas \ref{lem:rate} and \ref{lem:ub}, taken together, prove Theorem \ref{thm:main}. We thus obtain an exact characterization of the feedback capacity, in addition to providing an explicit coding scheme.
\section{Conclusion}
\label{sec:conclusion}
This work proposed explicit, deterministic coding schemes for the binary erasure channel (BEC) with a $(d,\infty)$-runlength limited (RLL) input constraint with feedback. In particular, a zero-error, labelling-based feedback capacity-achieving coding scheme was demonstrated, thereby allowing for the exact computation of the feedback capacity of the class of channels under consideration---a problem that was left open. It was also shown that feedback indeed increases the capacity of this class of input-constrained memoryless channels.


An important open question in the setting with feedback is whether the feedback capacity of general input-constrained discrete memoryless channels (DMCs) is achievable using input distributions of finite memory. As a first step, one could try to extend the coding schemes in the literature (including the one presented in this paper) to the BEC with a $(d,k)$-RLL input constraint, for $d\neq 0$ and $k\neq \infty$.


\appendices
\section{Proof of Proposition \ref{prop:simplify}}
Proposition \ref{prop:simplify} simplifies the maximizing values $\vec{\delta}^*:=(\delta_0^*,\ldots,\delta_d^*)$ in the two cases when the maximum is attained at an interior point of $\Delta_d$ and when it is attained on the boundary.

First, we characterize the stationary points of $R(\vec{\delta})$. We define
\begin{align*}
	N(\vec{\delta})&:=\bar{\epsilon}\left(\sum_{i=0}^{d}\epsilon^i h_b(\delta_i)\right),\ \text{and}\\
	D(\vec{\delta})&:=\sum_{i=0}^{d}\epsilon^i + d\bar{\epsilon}\left(\sum_{i=0}^{d}\epsilon^i\delta_i\right),
\end{align*}
to be the numerator and the denominator of $R(\vec{\delta})$, respectively. Note that for any $i\in [0:d]$,
\begin{align}
\frac{\partial R(\vec{\delta})}{\partial \delta_i}\Biggr|_{\vec{\delta} = \vec{\tilde{\delta}}} = 0 &\implies \bar{\epsilon}\epsilon^i\log_2\left(\frac{1-{\tilde{\delta}_i}}{\tilde{\delta}_i}\right)\cdot D(\vec{\tilde{\delta}}) = d\bar{\epsilon}\epsilon^i\cdot N(\vec{\tilde{\delta}}) \notag\\
&\implies \frac{N(\vec{\tilde{\delta}})}{D(\vec{\tilde{\delta}})} = \frac{1}{d}\log_2\left(\frac{1-\tilde{\delta_i}}{\tilde{\delta}_i}\right) \label{eq:maximizer}.
\end{align}
Therefore, from \eqref{eq:maximizer}, we get that
\[
\frac{\partial R(\vec{\delta})}{\partial \delta_i}\Biggr|_{\vec{\delta} = \vec{\tilde{\delta}}} = 0, \forall i\ \implies \tilde{\delta}_0 = \ldots = \tilde{\delta}_d.
\]
Thus, when the maximum in \eqref{eq:fbcap} is attained at an interior point, it holds that $\delta_0^* = \ldots =\delta_d^* =\delta^*$, thereby showing that in the first case, the feedback capacity
\[
C_{(d,\infty)}^{\text{fb}}(\epsilon) = \max_{\delta\in [0,\frac{1}{d+1}]} \frac{h_b(\delta)}{d\delta + \frac{1}{1-\epsilon}}.
\]
If the maximum is attained at a boundary point of $\Delta_d$, we wish to show that it holds that $\sum_{i=0}^{d}\delta_i^* = 1$. To this end, we first the define the following non-linear optimization problem with affine constraints:
\begin{align}
	&\text{\textbf{maximize}}\ R(\delta_0,\ldots,\delta_d) \notag\\
	&\text{\textbf{subj. to}}\ g_0(\vec{\delta}) := -\delta_0\leq 0,\ldots, g_d(\vec{\delta}) := -\delta_d\leq 0, \notag \\
	&\ \ \ \ \ \ \ \ \ \ \tilde{g}_0(\vec{\delta}) := \delta_0-1\leq 0,\ldots, \tilde{g}_d(\vec{\delta}) := \delta_d-1\leq 0, \notag\\
	&\ \ \ \ \ \ \ \ \ \ \ \ \ \ \ \ \ \ \hat{g}(\vec{\delta}) := \delta_0+\ldots+\delta_d - 1\leq 0. \label{eq:opt1}
\end{align}

Note that the objective function in \eqref{eq:opt1} and the constraint functions are all continuously differentiable in $[0,1]^{d+1}$ and the constraints are affine functions. Therefore, it holds from the necessity of the Karush-Kuhn-Tucker (KKT) conditions being satisfied, that there exist constants $\{\mu_i\}_{i=0}^{d}$, $\{\tilde{\mu}_i\}_{i=0}^{d}$, and $\hat{\mu}$, with $\mu_i, \tilde{\mu}_i \geq 0, \ \forall i$, and $\hat{\mu}\geq 0$, such that:
\begin{align}
	\label{eq:opt2}
\nabla R(\vec{\delta}) \Big|_{\vec{\delta} = \vec{\delta}^*} = \left(\sum_{i=0}^{d}\mu_i\nabla g_i(\vec{\delta})+\sum_{i=0}^{d}\tilde{\mu}_i\nabla \tilde{g}_i(\vec{\delta})+\hat{\mu}\nabla \hat{g}(\vec{\delta})\right)\Bigg|_{\vec{\delta} = \vec{\delta}^*}.
\end{align}
Suppose that the maximum is attained at a boundary point with $\delta_j^* = 0$, for some $j\in [0:d]$. Following reasoning similar to that in Lemma 13 in Appendix A of \cite{0k}, we note that since $\delta_j^* = 0$, we do not need to worry about the constraint that $\tilde{g}_j(\vec{\delta}) \leq 0$. Equation \eqref{eq:opt2} then gives us:
\[
\frac{\partial R(\vec{\delta})}{\partial \delta_j}\Biggr|_{\delta_j = 0} = -\mu_j + \hat{\mu}.
\]
However, we note that 
\[
\frac{\partial R(\vec{\delta})}{\partial \delta_j}
=\frac{\bar{\epsilon}\epsilon^j\log_2\left(\frac{1-{{\delta}_j}}{{\delta}_j}\right)\cdot D(\vec{\delta}) - d\bar{\epsilon}\epsilon^j\cdot N(\vec{\delta})}{\left(D(\vec{\delta})\right)^2},
\]
which tends to $+\infty$ as $\delta_j \to 0^+$. Therefore, it must be that $\hat{\mu}=+\infty$. 

Now, again, since the KKT conditions are necessary conditions for optimality in our maximization problem, from the complementary slackness condition, we see that:
\begin{align}
	\label{eq:opt3}
\hat{\mu}\hat{g}(\vec{\delta}^*)+\sum_{i=0}^{d}\mu_ig_i(\vec{\delta}^*)+\sum_{i=0}^{d}\tilde{\mu}_i\tilde{g}_i(\vec{\delta}^*) = 0.
\end{align}
From the fact that $\hat{\mu} = +\infty$ when the maximum is attained at a boundary point, it follows from equation \eqref{eq:opt3} that $\hat{g}(\vec{\delta}^*) = 0$, or, that $\sum_{i=0}^{d}\delta^*_i = 1$ when $\delta_j^*=0$ for some $j$.


\ifCLASSOPTIONcaptionsoff
  \newpage
\fi



%
\bibliographystyle{IEEEtran}
{\footnotesize
	\bibliography{references}}

%

%





\end{document}